\documentclass[12pt,oneside,letter,rotating,leqno]{article}
\usepackage{array}
\usepackage[english]{babel}
\usepackage{amsmath}
\usepackage{amssymb}
\usepackage{amsthm}
\usepackage{latexsym}
\usepackage[pdftex]{graphicx,color}
\usepackage{makeidx}
\usepackage{harvard}
\usepackage{endnotes}
\usepackage[nolists,tablesfirst]{endfloat}
\usepackage{multirow}
\usepackage{rotating}
\usepackage{exercise}
\usepackage{lscape}
\usepackage{verbatim}
\usepackage[small,bf]{caption}
\usepackage{chngpage}
\usepackage{booktabs}
\usepackage{verbatim}

\usepackage{longtable}
\usepackage{booktabs}

\usepackage{answers}

\newcommand{\R}{\mathbb{R}}
\newcommand{\prob}[1]{\mathbb{P}\left[#1\right]}

{}

\newcommand{\vbold}{\boldsymbol{v}}

\newcommand{\LL}{\mathbb{L}}
\newcommand{\N}{\mathbb{N}}

\newcommand{\X}{\mathcal{X}}

\newcommand{\abold}{\boldsymbol{\alpha}}
\newcommand{\bbold}{\boldsymbol{\beta}}
\newcommand{\lbold}{\boldsymbol{\lambda}}
\newcommand{\xbold}{\boldsymbol{x}}
\newcommand{\ebold}{\boldsymbol{e}}
\newcommand{\ubold}{\boldsymbol{u}}
\newcommand{\ybold}{\boldsymbol{y}}

\newcommand{\tbold}{\boldsymbol{\tau}}
\newcommand{\s}{\mathbb{S}}
\newcommand{\D}{\mathbb{D}}
\newcommand{\dbold}{\boldsymbol{d}}
\newcommand{\ineq}{\varepsilon}

\newcommand{\norm}[1]{\left\lVert#1\right\rVert}
\newcommand{\abs}[1]{\left\lvert#1\right\rvert}

\renewcommand{\baselinestretch}{1}

\begin{document}

\title{Naive Diversification Preferences and their Representation}

\author{{\bf Enrico G. De Giorgi}\footnote{Department of Economics, School of Economics and Political Science, University of St. Gallen, Bodanstrasse 6, 9000 St. Gallen, Switzerland, Tel. +41 71 2242430,
Fax. +41 71 224 28 94, email:
enrico.degiorgi@unisg.ch.} \and {\bf
Ola Mahmoud}\footnote{Faculty of Mathematics and Statistics, School of Economics and Political Science, University of St. Gallen, Bodanstrasse 6, 9000 St. Gallen, Switzerland and Center for Risk Management Research, University of California, Berkeley, Evans Hall, CA 94720-3880, USA, email: ola.mahmoud@unisg.ch}
}
\date{\today}
\maketitle

\linespread{1}
 \begin{abstract}
A widely applied diversification paradigm is the \emph{naive diversification} choice heuristic. It stipulates that an economic agent allocates equal decision weights to given choice alternatives independent of their individual characteristics. This article provides mathematically and economically sound choice theoretic foundations for the naive approach to diversification. We axiomatize naive diversification by defining it as a preference for equality over inequality and derive its relationship to the classical diversification paradigm. In particular, we show that (i) the notion of permutation invariance lies at the core of naive diversification and that an economic agent is a naive diversifier if and only if his preferences are convex and permutation invariant; (ii) Schur-concave utility functions capture the idea of being inequality averse on top of being risk averse; and (iii) the transformations, which rebalance unequal decision weights to equality, are characterized in terms of their implied turnover.

\mbox{}\\{\textbf{ Keywords}: naive diversification, convex preferences, permutation invariant preferences, Schur-concave utility, inequality aversion, majorization, Dalton transfer, Lorenz order. }\\
{\bf JEL Classification:} C02, D81, G11.
\end{abstract}
\renewcommand{\baselinestretch}{1}

\newtheorem{defn}{Definition}
\newtheorem{lemma}{Lemma}
\newtheorem{ass}{Assumption}
\newtheorem{prop}{Proposition}
\newtheorem{cor}{Corollary}
\newtheorem{thm}{Theorem}
\newtheorem{rem}{Remark}
\newtheorem{axiom}{Axiom}
\newtheorem{example}{Example}



\section{Introduction}
\label{section:intro}

Diversification is one of the cornerstones of decision making in economics and finance. In its essence, it conveys the idea of choosing variety over similarity. Informally, one might say that the goal behind introducing variety through diversification is the reduction of risk or uncertainty, and so one might identify a diversifying decision maker with a risk averse one. This is indeed the case in the expected utility theory (EUT) of \citeasnoun{VNM1944}, where risk aversion and preference for diversification are exactly captured by the concavity of the utility function which the decision maker is maximizing. However, this equivalence fails to hold in more general models of choice. We refer to \citeasnoun{DeGiorgiMahmoud2016} for a survey on the notion of diversification in the theory of choice.

In the context of portfolio construction, standard economic theory postulates that an investor  should optimize amongst various choice alternatives by maximizing portfolio return while minimizing portfolio risk, given by the return variance \cite{Markowitz1952}. In practice, however, these traditional optimization approaches to choice are plagued by technical difficulties.\footnote{These difficulties are  stemming from the instability of the optimization problem with respect to the available data. As is the case with any economic model, the true parameters are unknown and need to be estimated, hence resulting in uncertainty and estimation error. For a discussion of the problems arising in implementing mean-variance optimal portfolios, see  for example \citeasnoun{Hodges1978},  \citeasnoun{Best1991}, \citeasnoun{Michaud1998}, and \citeasnoun{Litterman2003}.} Experimental work in the decades after the emergence of the classical theories of \citeasnoun{VNM1944}  and \citeasnoun{Markowitz1952} has shown that economic agents in reality systematically violate the traditional diversification assumption when choosing among risky gambles. Indeed, seminal psychological and behavioral economics research by \citeasnoun{TverskyKahnemann1981} (see also \citeasnoun{Simon1955} and \citeasnoun{Simon1979})
suggests that the portfolio construction task may be too complex for decision makers
to perform. Consequently, investors adopt various types of simplified diversification
paradigms in practice.

One of the most widely applied such simple rules of choice is the so-called \emph{naive diversification} heuristic.  It stipulates that an economic agent allocates equal weights among a given choice set, independent of the individual characteristics of the underlying choice alternatives. In the context of portfolio construction, this rule is often referred to as the \emph{equal-weighted} or $1/n$ strategy.
This naive diversification paradigm goes as far back as the Talmud\footnote{The relevant Talmudic passage states that ``\emph{it is advisable for one that he should divide his money in three parts, one
of which he shall invest in real estate, one of which in business, and the third part to remain
always in his hands}."} and has been coined as \emph{Talmudic diversification} by \citeasnoun{Duchin2009}.
It is documented that even Harry Markowitz used the simple $1/n$ heuristic when he made his own retirement investments. He justifies his choice on psychological grounds: ``My intention was to minimize my future regret. So I split my contributions fifty-fifty between bonds and equities." \cite{Gigerenzer2010}

 \subsection{Related principles}

Naive diversification implies a preference of equality over inequality in the choice weights. One of the earliest, closely related hypotheses concerning decisions under subjective uncertainty is the \emph{principle of insufficient reason}, also called the \emph{principle of indifference}. It is generally attributed to \citeasnoun{Bernoulli} and invoked by \citeasnoun{Bayes} in his development of the binomial theorem. The principle states that in situations where there is no logical or empirical reason to favor any one of a set of mutually exclusive events or choices over any other, one should assign them all equal probability. In Bayesian probability, this is the simplest non-informative prior.

Outside the choice theoretic framework, the notion of preference of equality over inequality dominates several prominent problems in economic theory. Early in the twentieth century, economists became interested in measuring inequality of incomes or wealth. More specifically, it became desirable to determine how income or wealth distributions might be compared in order to say that one distribution was more equal than another. The first discussion of this kind was provided by \citeasnoun{lorenz1905}. He suggested a graphical manner in which to compare inequality in finite populations in terms of nested curves. If total wealth is uniformly distributed, the so-called \emph{Lorenz curve} is a straight line. With an unequal distribution, the curves will always begin and end in the same points as with an equal distribution, but they will be bent in the middle. The rule of interpretation, as he puts it, is: as the bow is bent, concentration increases. Later, \citeasnoun{dalton1920} described the closely related \emph{principle of transfers}. Under the theoretical proposition of a positive functional relationship between income and economic welfare, stating that economic welfare increases at an exponentially decreasing rate with increased income, Dalton concludes that maximum social welfare is achievable only when all incomes are equal. Following a suggestion by \citeasnoun{Pigou1912}, he proposed the condition that a transfer of income from a richer to a poorer person, so long as that transfer does not reverse the ranking of the two, will result in greater equity. Such an operation, involving the shifting of wealth from one individual to a relatively poorer individual, is known as the \emph{Pigou-Dalton transfer} and has also been labeled as a \emph{Robin Hood transfer}.
The seminal ideas of \citeasnoun{lorenz1905} and \citeasnoun{dalton1920} will be referenced frequently throughout our development of naive diversification preferences, as the mathematical framework upon which we rely coincides with theoretical formalizations of the Lorenz curve and the Dalton transfer.

\subsection{Experimental and empirical evidence of naive diversification}

Academics and practitioners have long studied the occurrence of naive diversification, along with its downside and potential benefits. Some of the first academic demonstrations of naive diversification as a choice heuristic were made by \citeasnoun{Simonson1990} in marketing in the context of consumption decisions by individuals, and by \citeasnoun{Read1995} in the context of experimental psychology. In the context of economic and financial decision making, empirical evidence suggests behavior which is consistent with naive diversification.
For instance, \citeasnoun{Benartzi2001} turned to study whether the effect manifests itself among investors making decisions in the context of defined contribution saving plans. Their experimental evidence suggests that some people spread their contributions evenly across the investment options irrespective of the particular mix of options. The authors point out that while naive diversification can produce a ``reasonable portfolio", it affects the resulting asset allocation and can be costly. In particular, people might choose a portfolio that is not on the efficient frontier, or they might pick the wrong point along the frontier. Moreover, it does not assure sensible or coherent decision making. Subsequently,
\citeasnoun{Huberman2006} find that participants tend to invest in only a small number of the
funds offered to them, and that they tend to allocate their contributions evenly across the funds that they use, with this tendency weakening with the number of funds used. More recently, \citeasnoun{Baltussen2011} find strong evidence for what they coin as \emph{irrational} behavior. Their subjects follow a conditional naive diversification heuristic as they exclude the assets with an unattractive marginal distribution and divide the available funds equally between the remaining, attractive assets. This strategy is applied even if it leads to allocations that are dominated in terms of first-order stochastic dominance -- hence the term irrational. Irrationality has been since then frequently used to describe naive diversification behavior. In \citeasnoun{Fernandes2013}, the naive diversification bias of \citeasnoun{Benartzi2001} was replicated across different samples using a within-participant manipulation of portfolio options. It was found that the more investors use intuitive judgments, the more likely they are to display the naive diversification bias.

In the context of portfolio construction, naive diversification has enjoyed a revival during the last few years because of its simplicity on one hand and the empirical evidence on the other hand suggesting superior performance compared to traditional diversification schemes. In addition to the relative outperformance, the empirical stability of the naive $1/n$ diversification rule has made it particularly attractive in practice, as --- unlike Markowitz's risk minimization strategies --- it does not rely on unknown correlation parameters that need to be estimated from data. Moreover, its outperformance has been investigated and a range of reasons have been proposed for why naive diversification may outperform other diversification paradigms. The most widely documented of these is the so-called small-cap-effect within the universe of equities. This theory stipulates that stocks with smaller market capitalization tend to ourperform larger stocks, and by construction, naive diversification gives more exposure to smaller cap stocks compared to capitalization weighting. Empirical support for the superior performance of equal weighted portfolios relative to capitalization weighting include \citeasnoun{Lessard1976}, \citeasnoun{Roll1981}, \citeasnoun{Ohlson1982}, \citeasnoun{Breen1989}, \citeasnoun{Grinblatt1989}, \citeasnoun{Korajczyk2004}, \citeasnoun{Hamza2007} and \citeasnoun{Pae2010}. Furthermore, \citeasnoun{Demiguel2007} show the strong performance relative to optimized portfolios. \citeasnoun{Duchin2009} provide a comparison of naive and Markowitz diversification and show that an equally weighted portfolio may often be substantially closer to the true mean variance optimality than an optimized portfolio. On the other hand, \citeasnoun{Tu2011} propose a combination of naive and sophisticated strategies, including Markowitz optimization, as a way to improve performance, and conclude that the combined rules not only have a significant impact in improving the sophisticated strategies, but also outperform the naive rule in most scenarios.

 \subsection{Towards choice-theoretic foundations}

The word \emph{naive} inherently implies a lack of sophistication. Indeed, our overview of naive diversification so far indicates that it is widely viewed as an anomaly linked to irrational behavior and that it does not assure sensible or coherent decision making. In its essence, the naive diversification paradigm is considered a simple and practical rule of thumb with no economic foundation guaranteeing its optimality. Moreover, despite the large experimental and empirical evidence of the presence and outperformance of naive diversification, a formalization of this heuristic within a choice theoretic or economic modelling framework does not seem to exist.

With the purpose of filling this gap, this paper provides mathematically and economically sound choice theoretic foundations for the naive approach to diversification of decision makers and investors. To this end, we axiomatize naive diversification by framing it as a choice theoretic preference for equality over inequality, which has a utility representation, and derive its relationship to the classical diversification paradigm. The crux of our choice theoretic aciomatization of the naive diversification heuristic lies in the idea that equality is preferred over inequality, a concept that is simultaneously simple and complex, as put by \citeasnoun{Sen1973}: ``\emph{At one level, it is the simplest of all ideas and has moved people with an immediate appeal hardly matched by any other concept. At another level, however, it is an exceedingly complex notion which makes statements of inequality highly problematic, and it has been, therefore, the subject of much research by philosophers, statisticians, political theorists, sociologists and economists.}" We complement this line of research from a decision theoretic perspective by using the mathematical concept of \emph{majorization}\footnote{Historically, majorization has been used to describe inequality orderings in the economic context of inequality of income, as developed by both \citeasnoun{lorenz1905} and \citeasnoun{dalton1920}. We refer the reader to \citeasnoun{MarshallOlkin2011} for a comprehensive self-contained account of the theory and applications of majorization.} to describe a preference relation which exhibits \emph{preference for naive diversification}.

The goal of our choice-theoretic approach is threefold: First, by developing an axiomatic system for what is considered to be an ``irrational" choice paradigm such as naive diversification, we can justify that it is in fact to some degree ``rational", in the sense that the axiomatization precisely captures widely observed regularities of behavior.
Second, this axiomatization enables us to gain some insights into the nature of the preferences and the utility of the naive diversifier that were previously unknown. In particular, by relating it to other known axiomatized behavioral paradigms, we will show that preferences for naive diversification are equivalent to convex preferences that additionally exhibit an indifference among the choice alternatives, which is formalized via a notion of permutation invariant preferences. This essentially implies that naive diversifiers simply have a different, yet consistent, set of preferences and utility functions that are closely related to, rather than contradicting, those of the traditional concave utility maximizers.
Finally, one may use the axioms underlying naive diversification to test the behavioral drivers of this choice heuristic in reality. For example, one of our axioms, that of permutation invariance, implies that the given alternatives are considered in some way symmetric or equivalent by the naive decision maker. This is an axiom that can be directly tested in, say, an experimental setting by relating it to Laplace's principle of indifference and varying the amount of information available  for each of the choice alternatives.
We will briefly revisit this last point in Section \ref{section:conclusion}.

The remainder of the paper is structured as follows. Section \ref{section:theory} sets up the choice theoretic framework and provides the necessary background on majorization and doubly stochastic matrices, both of which are fundamental concepts in our development. Section \ref{section:naive} presents an axiomatic formalization of naive diversification preferences and derive its relationship to the traditional (convex) diversification axiom. We then show that the notion of permutation invariance lies at the core of our definition and that a preference relation exhibits preference for naive diversification if and only if it is convex and permutation invariant. The corresponding utility representation in terms of Schur-concave functions is discussed in Section \ref{section:utility}. We show that Schur-concavity captures the idea of being inequality averse on top of being risk averse and discuss measures of inequality, which indicate how far from optimality a given choice allocation is and which allow for a quantitative comparison of two non-equal allocations in terms of their distance. Section \ref{section:rebalancing} characterizes the transformations which rebalance a choice allocation into another more equal allocation, in terms of the implied turnover and the induced transaction cost. In particular, we show that the least possible turnover is attained by applying the Pigou-Dalton transfer finitely many times. Section \ref{section:conclusion} concludes with a summary of our contributions and a brief discussion of possible generalizations of our development.


\section{Theoretical setup}\label{section:theory}

\subsection{Preference relation}

We consider a decision maker who chooses from the vector space $\X=\LL^\infty(\Omega,\mathcal{F},\mathbb{P})$ of essentially bounded real-valued random variables on a probability space $(\Omega, \mathcal{F}, \mathbb{P})$, where $\Omega$ is the set of states of nature, $\mathcal{F}$ is a $\sigma$-algebra of events, and $\mathbb{P}$ is a $\sigma$-additive probability measure on $(\Omega,\mathcal{F})$.\footnote{In this paper, we adopt the classical setup for risk assessment used in mathematical finance. However, almost all results presented in this paper also hold when alternative assumptions on $\X$ are made, e.g., $\X$ could be the space of probability distributions on a set of prizes, as often assumed in classical decision theory models.}

The decision maker is assumed to be able to form compound choices represented by the state-wise convex combination $\alpha\,x+(1-\alpha)\,y$ for $x,y\in\X$ and $\alpha\in[0,1]$, defined by $\alpha\,x(\omega) + (1-\alpha)\,y(\omega)$  for $\mathbb{P}$-almost all $\omega\in\Omega$. The space $\X$ is endowed with the order $x\ge y\Leftrightarrow x(\omega)\ge y(\omega)$ for $\mathbb{P}$-almost all $\omega\in\Omega$. For $n\geq 2$, $\abold = (\alpha_1,\dots,\alpha_n)\in\R^n$ and $(x_1, \dots, x_n)^\prime \in \X^n$, we will often denote the convex combination $\sum_{i=1}^n \alpha_i x_i$ by the dot product $\abold \cdot \xbold$.\footnote{For notational convenience, we write elements in $\R^n$ as row vectors and elements in $\mathcal{X}^n$ as column (random) vectors.} \\

A weak preference relation on $\X$ is a binary relation $\succsim$ satisfying:
\begin{enumerate}
\item[(i)] \emph{Completeness}: For all $x,y \in \X$, $x \succsim y \lor y \succsim x$.
\item[(ii)] \emph{Transitivity}: For all $x,y,z \in \X$, $x\succsim y \land y \succsim z \Rightarrow x \succsim z$.
\end{enumerate}

Every weak preference relation $\succsim$ on $\X$ induces an indifference relation $\sim$ on $\X$ defined by $x \sim y \Leftrightarrow (x \succsim y)\land (y \succsim x)$. The corresponding strict preference relation $\succ$ on $\X$ is defined by $x \succ y \Leftrightarrow x\succsim y \land \neg(x \sim y)$. A numerical or utility representation of the preference relation $\succsim$ is a real-valued function $u \colon \X \to \R$ for which $x \succsim y$ if and only if $ u(x) \geq u(y)$.

For $x\in\X$, $F_x$ denotes the cumulative distribution function of $x$, defined by $F_x(c) = \prob{x\leq c}$ for $c\in\R$, and $e(x)$ is the expected value of $x$, that is, $e(x) = \int x \ dF_x(x)$. For $c\in \mathbb{R}$, $\delta_c$ denotes the degenerated random variable with $\delta_c(\omega)=c$ for $\mathbb{P}$-almost all $\omega\in\Omega$. The \emph{certainty equivalent} of $x\in\X$ is the value $c(x)\in\R$ such that $x\sim\delta_{c(x)}$, i.e., $c(x)$ is the certain value which the decision maker views as equally desirable as a choice $x$ with uncertain outcome. The \emph{risk premium} $\pi(x)$ of $x\in\X$ is the amount by which the expected value of a choice $x\in\X$ must exceed the value of the guaranteed outcome in order to make the uncertain and certain choices equally attractive. Formally, it is defined as $\pi(x) = e(x)-c(x)$.

Emulating the majority of frameworks of economic theory, it seems reasonable to assume that decision makers prefer more to less. In particular, in view of the monetary interpretation of the space $\X$, a natural assumption on the preference relation $\succsim$ is \emph{monotonicity}.
\begin{enumerate}
\item[(iii)] \emph{Monotonicity}: For all $ x,y\in\X,$ $x\geq y \implies  x \succsim y$.
\end{enumerate}
Monotonicity of preferences is equivalent to having an increasing utility function $u$. Indeed, for $x\geq y$, we have $x\succsim y$ and thus $u(x)\geq u(y)$. Monotonicity of the utility function simply implies that an agent believes that ``more is better"; a larger outcome never yields lower utility, and for risky bets the agent would prefer a bet which is first-order stochastically dominant over an alternative bet.

Finally, continuity of preferences is assumed for technical reasons, as it can be used as a sufficient condition for showing that preferences on infinite sets can have utility representations. It intuitively states that if $x\succ y$, then small deviations from $x$ or from $y$ will not reverse the ordering.
\begin{enumerate}
\item[(iv)] \emph{Continuity}: For every $x\succ y$, there exist neighborhoods $B_x,B_y\subseteq \X$ around $x$ and $y$, respectively, such that for every $x'\in B_x$ and $y'\in B_y$, $x'\succ y'$.
\end{enumerate}
Throughout this article, unless otherwise stated, we assume that preferences are both monotonic and continuous. Debreu's theorem \cite{Debreu1964} states that there exists a continuous monotonic utility representation $u$ of a monotonic and continuous preference relation $\succsim$.

\subsection{Choice weights and majorization}

We use the theory of majorization from linear algebra to measure the variability of weights when diversifying across a set of $n$ possible choices. Majorization, which was formally introduced by \citeasnoun{hardy1934}, captures the idea that the components of a weight vector $\abold\in \R^n$ are less spread out or more nearly equal than the components of a vector $\bbold \in \R^n$. For any $\abold = (\alpha_1,\dots,\alpha_n)\in\R^n$, let
$$\alpha_{(1)} \geq \cdots \geq \alpha_{(n)}$$
denote the components of $\abold$ in \emph{decreasing} order, and let
$$ \abold_\downarrow = (\alpha_{(1)}, \dots, \alpha_{(n)}) $$
denote the decreasing rearrangement of $\abold$. The weight vector with $i$-th component equal to 1 and all other components equal to 0 is denoted by $\ebold_i$, and the vector with all components equal to 1 is denoted by $\ebold$.
We restrict our attention to non-negative weights which sum to one, that is, $\abold\in\s^n = \left\{ \vbold=(v_1,\dots,v_n)\in\R_+^n \mid \sum_{i=1}^n v_i = 1 \right\}$. This means that the  decision maker is assumed to use his full capital and is not taking ``inverse" positions such as shorting in financial economics.
Moreover, we will sometimes refers to the set
$$ \s_\downarrow^n = \left\{ \vbold_\downarrow=(v_{(1)},\dots,v_{(n)})\in\R_+^n \mid \sum_{i=1}^n v_{(i)} = 1 \right\} . $$

We now define the notion of majorization:

\begin{defn}[Majorization]
For $\abold= (\alpha_1,\dots,\alpha_n)\in\R^n$ and $\bbold=(\beta_1,\dots,\beta_n)\in\R^n$, $\bbold$ is said to \emph{(weakly) majorize} $\abold$ (or, equivalently, $\abold$ is \emph{majorized by} $\bbold$), denoted by $\bbold \ge_m \abold$, if
$$ \sum_{i=1}^n \alpha_i = \sum_{i=1}^n \beta_i $$
and for all $k=1,\dots,n-1$,
$$ \sum_{i=1}^k \alpha_{(i)} \leq \sum_{i=1}^k \beta_{(i)} \ . $$
\end{defn}

Majorization is a preorder on the weight vectors in $\s^n$ and a partial order on $\s_\downarrow^n$.
It is trivial but important to note that all vectors in $\s^n_\downarrow$ majorize the uniform vector $\ubold_n = (\frac{1}{n},\dots,\frac{1}{n})$, since the uniform vector is the vector with minimal differences between its components.

A key mathematical result in the study of majorization and inequality measurement is a theorem due to \citeasnoun{Hardy1929}. It roughly states that a vector $\abold$ is majorized by a vector $\bbold$ if and only if $\abold$ is an averaging of $\bbold$. This ``averaging" operation is formalized via doubly stochastic matrices.\footnote{A note on terminology: the term ``stochastic matrix" goes back to the large role that they play in the theory of discrete Markov chains. Doubly stochastic matrices are also sometimes called ``Schur transformations" or ``bistochastic".}
A square matrix $P$ is said to be stochastic if its elements are all non-negative and all rows sum to one. If, in addition to being stochastic, all columns sum to one, the matrix is said to be doubly stochastic. A formal definition follows.

\begin{defn}[Doubly stochastic matrix]
An $n\times n$ matrix $P=(p_{ij})$ is \emph{doubly stochastic} if $p_{ij}\geq0$ for $i,j=1,\dots,n$, $\ebold P = \ebold$ and $P\ebold ' = \ebold '$. We denote by $\D_n$ the set of $n\times n$ doubly stochastic matrices.
\end{defn}

\begin{thm}[\citeasnoun{Hardy1929}]
\label{theorem:hardy}
For $\abold,\bbold\in\R^n$, $\abold$ is majorized by $\bbold$ if and only if  $\abold = \bbold P$ for some doubly stochastic matrix $P$.\footnote{We refer the reader to \citeasnoun{Schmeidler1979} for several economic interpretations of Theorem \ref{theorem:hardy}, including decisions under uncertainty and welfare economics.}
 \end{thm}

An obvious example of a doubly stochastic matrix is the $n\times n$ matrix in which every entry is $1/n$, which we shall denote by $P_n$. Other simple examples are given by the $n\times n$ identity matrix $I_n$ and by permutation matrices: a square matrix $\Pi$ is said to be a permutation matrix if each row and column has a single unit entry with all other entries being zero. There are $n!$ such matrices of size $n\times n$ each of which is obtained by interchanging rows or columns of the identity matrix.  The set $\D_n$ of doubly stochastic matrices is convex and permutation matrices constitute its extreme points.

Use of a special type of doubly stochastic matrix, the so-called $T$-transform, will be made in this paper.

\begin{defn}[T-transform]
\label{defn:t_transform}
A (elementary) $T$-\emph{transform} is a matrix that has the form $T=\lambda I + (1-\lambda)\Pi$, where $\lambda\in[0,1]$ and $\Pi$ is a permutation matrix that interchanges exactly two coordinates. For $\abold=(\alpha_1,\dots,\alpha_n)\in\s^n$, $\abold T$ thus has the form
$$ \abold T = (\alpha_1, \dots, \alpha_{j-1}, \lambda\alpha_j+(1-\lambda)\alpha_k, \alpha_{j+1}, \dots, \alpha_{k-1}, \lambda\alpha_k+(1-\lambda)\alpha_j, \alpha_{k+1}, \dots, \alpha_n) \ , $$
where we assume that the $j$-th and $k$-th coordinates of $\abold$ are averaged.
\end{defn}

The importance of $T$-transforms can be seen from the following result, which is essential in the proof of Theorem \ref{theorem:hardy} and which we shall utilize in some of the proofs of this article.

\begin{prop}[\citeasnoun{Muirhead1903}; \citeasnoun{hardy1934}]
\label{prop:muirhead}
If $\abold\in\R^n$ is majorized by $\bbold\in\R^n$, then $\abold$ can be derived from $\bbold$ by successive applications of a finite number of $T$-transforms.
\end{prop}


\section{Naive diversification preferences}
\label{section:naive}

\subsection{Classical diversification}

An economic agent who chooses to diversify is traditionally understood to prefer variety over similarity. Axiomatically, preference for diversification is formalized as follows; see \citeasnoun{Dekel1989}.

\begin{defn}[Preference for diversification]
\label{defn:div}
A preference relation $\succsim$ exhibits \emph{preference for diversification} if for any $x_1, \dots, x_n \in \X$ and $\alpha_1, \dots, \alpha_n \in [0,1]$ for which $\sum_{i=1}^n \alpha_i = 1$,
$$ x_1 \sim \dots \sim x_n \Rightarrow \sum_{i=1}^n \alpha_i x_i \succsim x_j \quad \textrm{for all } j =1,\dots, n.$$
\end{defn}

This definition states that an individual will want to diversify among a collection of choices all of which are ranked equivalently. The most common example of such diversification is within the universe of asset markets, where an investor faces a choice amongst risky assets.

The related notion of convexity of preferences inherently relates to the classic ideal of diversification, as introduced by \citeasnoun{Bernoulli}.
By combining two choices, the decision maker is ensured under convexity that he is never ``worse off" than the least preferred of these two choices.

\begin{defn}[Convex preferences]
A preference relation $\succsim$ on $\X$ is \emph{convex} if for all $x,y \in \X$ and $\alpha\in(0,1)$,
$$ x  \succsim y \Rightarrow \alpha \,x + (1-\alpha)\, y \succsim y . $$
\end{defn}

Indeed, a monotonic and continuous preference relation is convex if and only if it exhibits preference for diversification, allowing us to use both axioms interchangeably. Moreover, it is well-known that a preference relation that is represented by a concave utility function is convex, and that a preference relation is convex if and only if its utility representation is quasi-concave.
Variations on this classical definition of diversification exist in the literature (see, for example, \citeasnoun{Chateauneuf2002} and
\citeasnoun{Chateauneuf2007}). We refer to \citeasnoun{DeGiorgiMahmoud2016} for a recent analysis of the classical definitions of diversification in the theory of choice.

\subsection{Naive diversification}

We present an axiomatic formalization of the notion of naive diversification in terms of preference of equal decision weights over unequal decision weights.

\begin{defn}[Preference for naive diversification]
\label{defn:naive_div}
A preference relation $\succsim$ exhibits \emph{preference for naive diversification} if for $n\in\N$, and $\abold=(\alpha_1,\dots,\alpha_n)\in\s^n$, $\bbold = (\beta_1, \dots, \beta_n) \in \s^n$ it follows that:
$$\abold \le_m \bbold \quad \Rightarrow \quad \sum_{i=1}^n \alpha_i x_i \succsim \sum_{i=1}^n \beta_i x_i \mbox{ for all $x_1,\dots,x_n\in\X$}.$$

A preference relation $\succsim$ exhibits \emph{preference for weak naive diversification} if for $n\in\N$ and $\abold=(\alpha_1,\dots,\alpha_n)\in\s^n$ it follows that:
$$ \frac{1}{n} \sum_{i=1}^n  x_i \succsim \sum_{i=1}^n \alpha_i x_i\mbox{ for all $x_1,\dots,x_n\in\X$}.$$
\end{defn}

This definition states that a preference relation $\succsim$ exhibits preference for naive diversification when a weight vector is always preferred to any alternative weight vector that majorizes it, i.e., it is more unequal; see \citeasnoun{Ibragimov2009}.
We now derive some initial properties of a preference relation $\succsim$ that exhibits preference for naive diversification:

\paragraph{(1) On naive versus weak naive diversification.} Definition~\ref{defn:naive_div} implies that $\frac{1}{n} \sum_{i=1}^n x_i \succsim \sum_{i=1}^n \alpha_i x_i $ for any $\abold \in \s^n$, because
any $\abold\in\s^n$ majorizes the equal-weighted decision vector $\ubold_n=\left(\frac{1}{n},\dots,\frac{1}{n}\right)$. It follows that the equal-weighted decision vector $\ubold_n$ is the most preferred choice allocation when $\succsim$ exhibits naive diversification preferences. This means that preference for naive diversification implies preference for weak naive diversification. However, the converse does not necessarily hold.

\paragraph{(2) On naive diversification and number of alternatives.} In general, we have
$$ \ubold_n\cdot\xbold \ \succsim\  \left(\frac{1}{n-1}, \dots, \frac{1}{n-1},0\right)\cdot\xbold \ \succsim \  \cdots \succsim \ \left(\frac{1}{2}, \frac{1}{2}, 0, \dots, 0\right)\cdot\xbold  \ \succsim \ \ebold_1 \cdot \xbold \ ,$$
for all $n\in\N$.
This ordering entails the informal diversification paradigm that \emph{more is better}, as analyzed by \citeasnoun{EltonGruber1977}, since an equal weighted allocation to $n$ choices is more preferred to an equal weighted allocation to $m$ choices if and only if $n\geq m$.

\paragraph{(3) On indifference under naive diversification.} Note that choice weights under naive diversification preferences are equivalent whenever their ordered vectors coincide. Moreover, whenever a collection of choices are pairwise equally ranked, a convex combination of each of these must be equally ranked. The following  formalization of these observations is hence an immediate consequence of Definition \ref{defn:naive_div}.

\begin{lemma}
\label{lemma:lemma1}
Let $\abold=(\alpha_1,\dots,\alpha_n)\in\s^n$, $\bbold=(\beta_1,\dots,\beta_n)\in\s^n$ and  $x_1, \dots, x_n, y_1, \dots, y_n \in \X$ such that $x_i \sim y_i$ for $i=1,\dots,n$. Suppose that $\succsim$ exhibits preference for naive diversification. Then
\begin{enumerate}
\item[(i)] $\sum_{i=1}^n \alpha_i x_i \sim \sum_{i=1}^n \beta_i x_i$ if and only if $\sum_{i=1}^k \alpha_{(i)}= \sum_{i=1}^k \beta_{(i)}$ for all $k=1,\dots,n$;
\item[(ii)] $ \sum_{i=1}^n \alpha_i x_i \sim \sum_{i=1}^n \alpha_iy_i $.
\end{enumerate}
\end{lemma}

\paragraph{(4) On naive diversification and convex preferences.} An agent whose preferences are convex chooses to diversify by taking a convex combination over individual choices without specifying a preference ordering over choice weights. So the classical notion of diversification does not necessarily imply preferences for naive diversification. The converse holds however: suppose that $\succsim$ exhibits preferences for naive diversification and let $x_1,\dots,x_n\in\X$ with $x_1\sim \dots\sim x_n$. Then, for $\abold=(\alpha_1,\dots,\alpha_n)\in\s^n$, we have $\sum_{i=1}^n \alpha_i x_i\succsim x_j$ for all $j=1,\dots,n$, since the components of the choice vector $\abold$ are more nearly equal than those of $\ebold_j$, i.e., any $\abold\in \s^n$ is majorized by $\ebold_j$. This proves the following result.
\begin{prop}
\label{prop:naive_convex}
Naive diversification preferences are convex, or, equivalently, exhibit preferences for diversification.
\end{prop}

\subsection{Permutation invariant preferences}

The notion of permutation invariance lies at the core of the definition of naive diversification. Permutation invariance captures the idea that the underlying characteristics of the individual choices are irrelevant in the decision making process. In other words, the economic agent is indifferent towards a permutation of the components of choice vectors. We formalize such permutation invariant preferences through permutation matrices. For a permutation matrix $\Pi$ and choice vector $\abold=(\alpha_1,\dots,\alpha_n)\in\s^n$, we shall write $\abold\Pi$ for the vector whose components have been shuffled using $\Pi$ and whose $i$-th component we denote by $(\abold\Pi)_i$. When ordering the components of $\abold\Pi$ in decreasing order, we denote its $i$-th ordered component by $(\abold\Pi)_{(i)}$.

\begin{defn}[Permutation invariant preferences]
A preference relation $\succsim$ on $\X$ is \emph{permutation invariant} if for all $x_1, \dots, x_n\in\X$ and $\abold = (\alpha_1, \dots, \alpha_n) \in \s^n$,
$$ \abold \cdot \xbold \sim (\abold\Pi)\cdot \xbold \ , $$
where $\Pi$ is a permutation matrix and $\xbold = (x_1,\dots,x_n)$.
\end{defn}

The following lemma shows that naive diversification preferences are permutation invariant.

\begin{lemma}
\label{lemma:naive_sym1}
Naive diversification preferences are permutation invariant.
\end{lemma}
\begin{proof}
For all $\abold = (\alpha_1,\dots,\alpha_n) \in\s^n$, we have $\abold_\downarrow = (\abold\Pi)_\downarrow$. Therefore, $\sum_{i=1}^k \alpha_{(i)} = \sum_{i=1}^k (\abold\Pi )_{(i)}$ for all $k=1,\dots,n$. By Lemma \ref{lemma:lemma1}, this implies that $ \abold \cdot \xbold \sim (\abold\Pi)\cdot \xbold$.
\end{proof}

Recall that Lemma \ref{lemma:lemma1} characterizes indifference between choice weights in terms of equality of the corresponding ordered vectors. An immediate consequence is the following Corollary, which states that an indifference between two choice allocations does not necessarily correspond to the choice vectors being equal, but that they differ in terms of permutation.

\begin{cor}
\label{cor:naive_sym}
Let $\abold,\bbold\in\s^n$ and  $x_1, \dots, x_n\in \X$. If $\succsim$ exhibits preference for naive diversification, then
$\abold = \bbold\Pi$
if and only if
$ \sum_{i=1}^n \alpha_i x_i \sim \sum_{i=1}^n \beta_ix_i $.
\end{cor}

The significance of permutation invariance manifests itself in its implication for classical diversification. Indeed, imposing permutation invariance on convex preferences yields preferences for naive diversification (Proposition \ref{prop:naive_sym2}). We start by showing the weaker result.

\begin{prop}
\label{prop:naive_sym1}
A preference relation $\succsim$ that is permutation invariant and convex exhibits preference for weak naive diversification.
\end{prop}
\begin{proof}
Because any $\abold =(\alpha_1,\dots,\alpha_n)\in\s^n$ majorizes the vector $\bold{u}_n$, then, according to Proposition~\ref{prop:muirhead}, $\bold{u}_n$ can be derived from $\abold$ by successive applications of a finite number of $T$-transforms, i.e.,
$$\bold{u}_n=\abold T_1T_2\cdots T_k$$
where $T_1,T_2,\cdots T_k$ are $T$-transforms. For $x_1,\dots,x_n\in\X$, we have:
$$\frac{1}{n}\,\sum_{i=1}^n x_i=\bold{u}_n\cdot \xbold=(\abold T_1\cdots T_k )\cdot \xbold.$$

We prove that $(\abold T_1\cdots T_k )\cdot \xbold\succsim\abold\cdot\xbold$ by mathematical induction. First of all, we show that $(\abold T)\cdot \bold{x}\succsim \abold \cdot \xbold$ when $T$ is $T$-transform and $\succsim$ is permutation invariant and convex. Indeed,
$$(\abold T)\cdot\xbold=[\abold (\lambda\,I+(1-\lambda)Q)]\cdot\xbold=\lambda\,\abold\cdot\xbold+(1-\alpha)\,(\abold Q)\cdot\xbold$$
where $Q$ is a permutation matrix. Because $\succsim$ is permutation invariant, then $(\abold Q)\cdot \xbold\sim \abold\cdot\xbold$. Finally, because $\succsim$ is convex, then
$$\lambda\,\abold\cdot\xbold+(1-\lambda)\,(\abold Q)\cdot \bold{x}\succsim \abold \cdot\xbold.$$
It follow that
$$(\abold T)\cdot \bold{x}\succsim \abold \cdot \xbold.$$
Now suppose that $(\abold T_1\cdots T_{k-1} )\cdot \xbold\succsim\abold\cdot\xbold$. Let $\tilde{\abold}=\abold T_1\cdots T_{k-1}$. It follows that:
$$(\abold T_1\cdots T_k )\cdot \xbold =(\tilde{\abold} T_k)\cdot \xbold \succsim \tilde{\abold} \cdot \xbold =(\abold T_1\cdots T_{k-1} )\cdot \xbold\succsim\abold\cdot\xbold.$$
Therefore,
$$(\abold T_1\cdots T_k )\cdot \xbold\succsim \abold\cdot\xbold.$$
This proves the statement of the proposition.\end{proof}

We recall that $T$-transforms (Definition \ref{defn:t_transform}) are averaging operators between two components of the original weight vector. This averaging operator is always weakly preferred under permutation invariant and convex preferences. The proof of Proposition \ref{prop:naive_sym1} shows that repeated averaging of two components of a weight vector reaches its limit at the equal-weighted decision vector $\ubold_n$. Therefore, Proposition \ref{prop:naive_sym1} can be viewed as a corollary to Muirhead's result (Proposition \ref{prop:muirhead}).

Another seminal result tangentially related to Proposition \ref{prop:naive_sym1} appeared in \citeasnoun{Samuelson1967}, where the first formal proof of the following, at the time seemingly well-understood, diversification paradigm is given: ``\textit{putting a fixed total of wealth equally into independently, identically distributed investments will leave the mean gain unchanged and will minimize the variance}.'' One may hence think of the conditions of having non-negative, independent and identically distributed random variables in Theorem 1 of \citeasnoun{Samuelson1967} being replaced by the permutation invariance condition in Proposition \ref{prop:naive_sym1} to yield an equal weighted allocation as optimal.\footnote{See \citeasnoun{HadarRussell1969}, \citeasnoun{HadarRussell1971}, \citeasnoun{Tesfatsion1976} and \citeasnoun{LiWong1999} for generalizations of Samuelson's classical result.}

We next derive the stronger statement, which gives naive diversification under permutation invariance and convexity.

\begin{prop}
\label{prop:naive_sym2}
A preference relation $\succsim$ that is permutation invariant and convex exhibits preference for naive diversification.
\end{prop}
\begin{proof}
Suppose that $\succsim$ is permutation invariant and convex. We have to show that $\abold\cdot\xbold\bold\succsim \bbold\cdot \xbold$ for all $\xbold^\prime\in\X$ when $\bbold\ge_m\abold$. If $\bbold\ge_m\abold$, then $\abold$ can be derived from $\bbold$ by successive applications of a finite number of $T$-transforms. By applying the same argument as in the proof of Proposition~\ref{prop:naive_sym1}, we have $\abold\cdot\xbold \succsim \bbold\cdot \xbold$. Therefore, $\succsim$ exhibits preference for naive diversification.
\end{proof}
Combining Proposition \ref{prop:naive_convex} and Lemma \ref{lemma:naive_sym1} with Proposition \ref{prop:naive_sym2} yields the following equivalence of preferences.

\begin{thm}
A monotonic and continuous preference relation $\succsim$ exhibits preference for naive diversification if and only if it is convex and permutation invariant.
\end{thm}


\section{Representation}
\label{section:utility}

\subsection{Utility representation}

We consider a preference relation $\succsim$ that exhibits preference for naive diversification and also possesses a utility representation $u$, i.e., $x\succsim y$ if and only if $u(x)\ge u(y)$. In this section we study the implications of preference for naive diversification (or, equivalently, permutation invariance and convexity) on the utility function $u$. Without loss of generality we restrict our analysis to the case where diversification is among at most $n$ choice alternatives. Let us fix $\xbold=(x_1,\dots,x_n)^\prime\in\X^n$ and for $\abold\in\s^n$ define
$$\phi_{\xbold}(\abold)=u(\abold\cdot\xbold).$$
Because under naive diversification $\abold\cdot \xbold\succsim \bbold\cdot \xbold$ when $\bbold \ge_m\abold$, then
$$\phi_{\xbold}(\abold)=u(\abold\cdot\xbold)\ge u(\bbold\cdot\xbold)=\phi_{\xbold}(\bbold)$$
when $\bbold\ge_m\abold$. Therefore, $\phi_{\xbold}$ is Schur-concave.\footnote{A function $f\,:\,\R^n\rightarrow\R$ is said to be Schur-concave when $\xbold$ is majorized by $\ybold$ implies $f(\xbold)\ge f(\ybold)$.} The utility function $u\,:\,\X\rightarrow \R$, on the other hand, is generally \emph{not} Schur-concave on $\X=\LL^\infty(\Omega,\mathcal{F},\mathbb{P})$. Indeed, $u$ is Schur-concave on $\X=\LL^\infty(\Omega,\mathcal{F},\mathbb{P})$ when it is consistent with the convex order on $\LL^\infty(\Omega,\mathcal{F},\mathbb{P})$; see \citeasnoun{Dana2005}. Proposition 2.1 in \citeasnoun{Dana2005} shows that $u$ is monotone and Schur-concave on $\LL^\infty(\Omega,\mathcal{F},\mathbb{P})$ if and only if it is consistent with second-order stochastic dominance. However, a preference relation that exhibits preference for naive diversification is monotone but generally not consistent with second-order stochastic dominance, which implies that $u$ is generally not Schur-concave on $\X$.

We point out that in general Schur-concave functions are neither concave nor quasi-concave. However, in our case, because the preference relation is convex (as naive diversification implies permutation invariance and convexity), both $u$ and $\phi_{\xbold}$ are quasi-concave. A seminal result known as the Schur-Ostrowski criterion (due to \citeasnoun{Schur1923} and \citeasnoun{Ostrowski1952}) is that when $\phi_{\xbold}$ is differentiable then it is Schur-concave if and only if it is symmetric (i.e., $\phi_{\xbold}(\abold)=\phi_{\xbold}(\abold Q)$ for any permutation $Q$) and
$$(\alpha_i - \alpha_j)\,\left(\frac{\partial \phi_{\xbold}(\abold)}{\partial \alpha_i} - \frac{\partial \phi_{\xbold}(\abold)}{\partial \alpha_j}\right) \le 0 $$
for all $i\neq j$. This implies that $\frac{\partial \phi_{\xbold}}{\partial \alpha_i}(\abold)\le \frac{\partial \phi_{\xbold}(\abold)}{\partial \alpha_j}$ when $\alpha_i>\alpha_j$. In other words, utility gained from
naive diversification is more sensitive to changes in smaller component weights. Increasing smaller choice components (those that are $\leq \frac{1}{n}$) implies moving closer to equality and hence leads to higher utility, whereas an increase in the relatively larger components (those that are already $\geq \frac{1}{n}$) decreases utility as the resulting choice vector is further away from the equal weighted vector. 

Well-known concrete examples of functions that could represent naive diversification preferences are the standard deviation and entropy functions, which we illustrate for the sake of simplicity on $\s^n$:
\begin{example}
\label{example:utility}
\begin{enumerate}
\item[(i)] The \emph{standard deviation} $\sigma:\s^n\to\R$ defined by $$\sigma(\lambda_1,\dots,\lambda_1) = \left[ \frac{1}{n} \sum_{i=1}^n (\lambda_i-\frac{1}{n})^2 \right]^{\frac{1}{2}}$$ is Schur-convex, which means that its negative $-\sigma$ is Schur-concave. \citeasnoun{dalton1920} was the first to study standard deviation as a measure of inequality of incomes.
\item[(ii)] The \emph{Shannon information entropy} $H:\s^n\to\R$ defined by $$H(\lambda_1,\dots,\lambda_n) = - \sum_{i=1}^n \lambda_i \log \lambda_i$$ represents naive diversification preferences and is commonly used as a measure of inequality in a population.\footnote{The term entropy in general implies disorder, unpredictability or uncertainty. In information theory, entropy is used to measure uncertainty of outcomes, and maximum uncertainty is reached when all outcomes are equally probable and hence difficult to predict. Entropy is also commonly used to measure diversity in ecological, biological and information sciences.}
\end{enumerate}
\end{example}

As shown in Section~\ref{section:naive}, naive diversifiers are \emph{inequality averse} on top of being risk averse. Schur-concavity (and Schur-convexity) have historically been the key notions in inequality theory. They are equivalent to requiring monotonicity with respect to the Lorenz order \cite{lorenz1905} and the Dalton principle of transfers \cite{dalton1920}. Our work supports these findings by showing that Schur-concavity captures the notion of being inequality averse in the context of choice under uncertainty, and its monotonicity property extends to orders that exhibit preference for naive diversification.

Moreover, \citeasnoun{Schur1923} and \citeasnoun{Hardy1929} showed that a vector $\abold=(\alpha_1,\dots,\alpha_n)\in\s^n$ is closer to equality than $\bbold=(\beta_1,\dots,\beta_n)\in\s^n$ if and only if $\sum_{i=1}^n g(\alpha_i)\geq \sum_{i=1}^n g(\beta_i)$ for all continuous concave functions $g:[0,1]\to\R$.
This condition can be interpreted as a ranking for aggregating the utility of individual choices --- if an agent is a concave risk averse utility maximizer, the aggregate total utility of $n$ individual choices must be higher for an allocation with less variation, when the agent is inequality averse on top of being risk averse.\footnote{This is a well-known condition in social choice theory, where it corresponds to the social welfare ranking: $g$ stands for the individual utility function of every individual $x_i$ in a population of size $n$, and $\sum_{i=1}^n g(x_i)$ is the overall social welfare utility of the population.}

\subsection{Measuring naive diversification.}

We briefly discuss properties that measures of the degree of Schur-concave optimality should satisfy.
An evaluation of the optimality of a given choice allocation of a naive diversifier essentially reduces to a measure of inequality of the decision weights of his choice. Measures of inequality arise in various disciplines within economic theory, particularly within the context of wealth and income.\footnote{There is a vast literature on diversity and inequality indices in economics --- see classical discussions and surveys by \citeasnoun{Sen1973}, \citeasnoun{Szal1977}, \citeasnoun{dalton1920}, \citeasnoun{Atkinson1970}, \citeasnoun{Blackorby1978}, and \citeasnoun{Kraemer1998}.} Most of these indices have been developed primarily based on foundations of the concept of social welfare, and hence may not necessarily be applicable to our setting.
Since a measure of inequality strongly depends on the context, we provide an axiomatization that is consistent with our definition of preference for naive diversification, which has a precise mathematical formulation in terms of majorization and Schur-concave functions. Many indices are qualitative in nature focused on ranking with no indication of a quantification of the comparison.
We do not only seek a qualitative ranking of choice allocations, but we aim to quantify the distance between two weight allocations. The resulting measure hence indicates how far from optimality a given choice allocation is and allows for comparison of two non-equal choice allocations in terms of their distance.

Let $\succsim$ be a preference relation on $\X$ exhibiting preferences for naive diversification.
To derive the qualitative and quantitative properties that are consistent with naive diversification, we fix the optimal choice allocation $\ubold_n=(\frac{1}{n},\dots,\frac{1}{n})$ for a given $n$ and look at comparisons with respect to this vector. The following are the minimal requirements that a measure $\mu_n:\X\to\R$ of naive diversification should satisfy:

\begin{itemize}
\item[(A1)] Positivity: For all $x\in\X$, $\mu_n(x) \geq 0$.
\vspace{-10pt}
\item[(A2)] Normality: For all $x\in\X$, $\mu_n(x) = 0$ if and only if $x\sim \ubold_n\cdot\xbold$ for some $\xbold=(x_1,\dots,x_n)$.

\vspace{-10pt}
\item[(A3)] Boundedness: For all $x\in\X$, $\mu_n(x)<\infty$.
\vspace{-10pt}
\item[(A4)] Representation: For all $x,y\in\X$, $x\succsim y$ implies $\mu_n(x)\geq\mu_n(y)$.
\vspace{-10pt}
\item[(A5)] Permutation invariance: For $\abold=(\alpha_1,\dots,\alpha_n), \bbold=(\beta_1,\dots,\beta_n)\in\s^n$ and $x_1,\dots,x_n\in\X$, if $\sum_{i=1}^n\alpha_i x_i \succsim \sum_{i=1}^n \beta_i x_i$ and $\abold\neq\Pi\bbold$ for any permutation matrix $\Pi$, then $\mu_n(\sum_{i=1}^n\alpha_i x_i ) > \mu_n(\sum_{i=1}^n\beta_i x_i )$.
\end{itemize}

Axioms A1, A2, and A3 essentially ensure that the function $\mu_n$ is a well-behaved probability metric \cite{Rachev2011} and hence an analytically sound measure of the distance between two random quantities. Axiom A4 implies Schur-concavity and thus that the qualitative ranking is preserved. By introducing invariance under permutation (Axiom A5), we require strict Schur-concavity. This distinguishes equivalence, and hence a zero distance from equality, from a strict preference ordering of choice weights, which should give a strictly positive distance.

\begin{example} The following well-known classes of measures from statistics, economics and asset management satisfy the above axioms and are hence in particular Schur-concave (Axiom A4) and invariant under permutation (Axiom A5):
\begin{enumerate}
\item[(i)] \emph{Statistical dispersion measures.} Many familiar measures of dispersion, or their negation, satisfy the above axioms. Examples include the variance (as opposed to the standard deviation from Example \ref{example:utility}); its normalized version, which is known as the coefficient of variation; and the entropy function from Example \ref{example:utility}. Other closely related functions, such as $g(\sum_{i=1}^n\alpha_ix_i) = [\frac{1}{n} \sum_{i=1}^n(\log \alpha_i )^2]^{1/2}$, are sometimes used to measure inequality, but do not satisfy these two axioms, however.
\item[(ii)] \emph{Economic inequality indices.} The Lorenz curve has suggested several indices aimed to measure economic equality, and they can be shown to satisfy Axioms A4 and A5.  One of the most widely used such measure is the Gini coefficient \cite{Gini1921}, or the related \emph{Gini mean difference}, defined by
$ G(\sum_{i=1}^n\alpha_ix_i) = \frac{1}{n^2}\sum_j\sum_k \abs{\alpha_j-\alpha_k} $.
Other examples are given by functions that are sums of strictly concave utilities, such as Dalton's measure \cite{dalton1920} and Atkinson's measure \cite{Atkinson1970}.

\item[(iii)] \emph{Diversification indices.} Many measures of portfolio concentration in finance are referred to as \emph{diversification} indices, whereas they are in fact measures of \emph{naive diversification}. One of the most standard such measure is the \emph{Herfindahl-Hirschman Index} \cite{Hirschman1964}: $H(\sum_{i=1}^n\alpha_ix_i) =( \sum_{i=1}^n \alpha_i^2 -1/n)/(1-1/n)$.
This measure is equivalent to the \emph{Simpson diversity index} \cite{Simpson1949}, which is simply a sum of squares of the choice weights.
\end{enumerate}
\end{example}


\section{Rebalancing to equality}
\label{section:rebalancing}

Based on Theorem \ref{theorem:hardy} of \citeasnoun{Hardy1929}, a doubly stochastic matrix can be thought of as an operation between two weight allocations leading towards greater equality in the weight vector. With this in mind, we define a \emph{rebalancing transform} to be a doubly stochastic matrix. Clearly, rebalancing in this context cannot yield a less diversified allocation. In other words, applying a rebalancing transform to a vector of decision weights is equivalent to averaging the decision weights.

In this Section, we characterize such transforms which start with a suboptimal weight allocation $\sum_{i=1}^n\alpha_ix_i$ and produce equality $\frac{1}{n}\sum_{i=1}^nx_i$ in terms of their implied turnover in practice. Our analysis is focused on the asset allocation problem, where rebalancing is understood in terms of buying and selling positions. However, this discussion can be generalized to characterize transforms in the context of reallocation of wealth, such as Dalton's principle of transfers.

\subsection{Rebalancing polytopes}
Starting from an allocation $\abold\in\s^n$, there are, in general, more than one possible transforms that rebalance $\abold$ to $\ubold_n$ or, more generally, to an allocation $\bbold\in\s^n$ that is closer to equality. Given two weight allocations $\abold,\bbold\in\s^n$ with $\abold$ majorized by $\bbold$, the set
$$\Omega_{\abold\le_m \bbold} = \{ P\in\D_n \mid \abold = \bbold P \}$$
is referred to as the \emph{rebalancing polytope} of the order $\abold\le_m\bbold$.\footnote{Within the linear algebra literature, this set is referred to as the ``majorization polytope". As pointed out by \citeasnoun{MarshallOlkin2011}, very little is known about this polytope.} The set $\Omega_{\abold\le_m\bbold}$ is nonempty, compact and convex. In the case that the components of $\bbold$ are simply a rearrangement of the components of $\abold$, then $\Omega_{\abold\le_m\bbold}$ contains one unique permutation matrix. In general, however, $\Omega_{\abold\le_m\bbold}$ contains more than one element.

Now, for $\lbold\in\s^n$, we have $\ubold_n\le_m\lbold$, and so our focus henceforth is the set
$$\Omega_{n,\lbold} :=\Omega_{\ubold_n\le_m\lbold}= \{ P\in\D_n \mid \ubold_n=\lbold P \}\ .$$
It contains all rebalancing transformations that lead to an equal allocation. In particular, it includes the matrix $P_n$ with all entries equal to $1/n$.

\begin{example}
\label{example:rebalance}
Let $n=3$, so $\ubold_3 = (\frac{1}{3},\frac{1}{3}, \frac{1}{3})$, and let $\lbold = (\frac{1}{2},\frac{1}{3},\frac{1}{6})$. Denote the entries of a 3-dimensional matrix $P$ by $p_{ij}$, for $i,j=1,2,3$. Then, since for all $j=1,2,3$, $\sum_{i=1}^3 \lambda_ip_{ij}=\frac{1}{3}$ and because the columns of $P$ add to 1, we have
$ \frac{1}{2}p_{11} + \frac{1}{3}p_{21} + \frac{1}{6}p_{31} = \frac{1}{3} $ and
$ \frac{1}{2}p_{12} + \frac{1}{3}p_{22} + \frac{1}{6}p_{32} = \frac{1}{3} $.
Letting $p_{12} = u$ and $p_{21}=v$, rebalancing transforms $P= P(u,v) \in\Omega_{n,\lbold}$ take the form
$$ P(u,v) =  \begin{pmatrix}
\frac{1}{2}-\frac{1}{2}v & u & \frac{1}{2} + \frac{1}{2}v - u\\[0.4em]
v & 1-2u & 2u-v \\[0.4em]
\frac{1}{2}-\frac{1}{2}v & u & \frac{1}{2} + \frac{1}{2}v - u
\end{pmatrix} \ .$$
Being a doubly stochastic matrix, the entries of $P$ lie within $[0,1]$, and so the feasible region lies in the first quadrant with respect to the variables $u$ and $v$. For example, for $u=v=0$, we have
$$ P(0,0) =  \begin{pmatrix}
\frac{1}{2} & 0 & \frac{1}{2}\\[0.4em]
0 & 1 & 0 \\[0.4em]
\frac{1}{2} & 0 & \frac{1}{2}
\end{pmatrix} \ .$$
Note that for $u=v=\frac{1}{3}$, we obtain $P(\frac{1}{3},\frac{1}{3}) = P_3$.
\end{example}

\subsection{Rebalancing turnover}

As in Example \ref{example:rebalance}, we are interested in rebalancing a weight allocation towards equality in practice. However, it is not clear how or why one would choose one transform in a given polytope $\Omega_{n,\lbold}$ over another. We provide a precise distinction in terms of \emph{turnover}. In the context of asset allocation, the particular rebalancing transform applied to rebalance one weight allocation to another has an interpretation in terms of the fraction of assets bought and sold and, consequently, in terms of the implied transaction costs.

\begin{defn}[Turnover]
\label{defn:turnover}
For $\lbold\in\s^n$, the  \emph{turnover vector} $\tbold(\lbold)$ corresponding to rebalancing $\lbold$ to equality $\ubold_n$ is given by $\tbold(\lbold) = \lbold-\ubold_n$, and the resulting \emph{turnover} $\tau(\lbold)$ is defined by $\tau(\lbold) =  \frac{1}{2} \sum_{i=1}^n \abs{\tau_i} $, where $\tau_i$ are the components of the turnover vector $\tbold(\lbold)$.
\end{defn}

The turnover is intuitively equal to the portion of the total decision weights that would have to be redistributed by taking from weights exceeding $1/n$ and assigning these portions to weights that are less than $1/n$. The turnover hence always lies between 0 and 1. Graphically, it can be represented as the longest vertical distance between the Lorenz curve associated with a choice vector, and the diagonal line representing perfect equality. Note the similarities between Definition \ref{defn:turnover} and the \emph{Hoover Index} \cite{Hoover1936}, a measure of income metrics which is also known as the Robin Hood Index, as uniformity is achieved in a population by taking from the richer half and giving to the poorer half.

\begin{lemma}
Let $\lbold\in\s^n$ and $\Omega_{n,\lbold} = \{ P\in\D_n \mid \ubold_n=\lbold P \}$. Then for all $P\in\Omega_{n,\lbold}$,
$$ \lbold (I_n-P) = \tbold(\lbold) \ . $$
\end{lemma}

\begin{proof}
The equation follows by definition, as $\lbold(I_n-P) = \lbold - \lbold P = \lbold - \ubold_n = \tbold(\lbold)$.
\end{proof}

\begin{example}
The turnover resulting from transforming $\lbold = (\frac{1}{2},\frac{1}{3},\frac{1}{6})$ of Example \ref{example:rebalance} to equality $\ubold_3$ is equal to $1/6$, which means that about 16.67\% of a given portfolio's assets would theoretically have to be sold and bought to obtain an equal allocation.
\end{example}

Based on Definition \ref{defn:turnover}, every transformation $P\in\Omega_{n,\lbold}$ applied to $\lbold$ theoretically yields the same turnover. However, there is a subtle difference.
Consider once again Example \ref{example:rebalance}, where one could apply $P(0,0)$ and $P_3$ to rebalance the allocation $\lbold = (\frac{1}{2},\frac{1}{3},\frac{1}{6})$ to equality. Note that the second entry of $\lbold$ can remain as is, and one needs only to average its first and third entries to obtain equality. This is precisely the transformation $P(0,0)$: the matrix entry $p_{22}=1$ guarantees that $\lambda_2 = 1/3$ remains as is, whereas the remaining non-zero matrix entries average out the first and third entries of $\lbold$. On the other hand, the transformation $P_3$ takes the average of each of the three entries of $\lbold$. In practice, this transformation would imply that the actual turnover is higher than the theoretical turnover of 1/6. This is because more assets are bought or sold than is theoretically needed to obtain equality. In our simple example of 3 possible choices, choosing a transformation that minimizes turnover is straightforward. However, for larger collections, the choice of the optimal rebalancing transformation may not be obvious.

We refer to the actual turnover induced in practice as the \emph{practical turnover}.

\begin{defn}[Practical turnover]
Let $\lbold\in\s^n$. For $P\in\Omega_{n,\lbold}$, the \emph{practical turnover} is given by $\tilde{\tau}_P(\lbold) = \tau(\lbold) \norm{P-I_n}$, where $\norm{\cdot}$ is the Frobenius norm taken up-to-permutation.\footnote{For a $m\times n$ matrix $A=(a_{ij})$, the Frobenious norm is defined as $\norm{A}=\sqrt{\sum_{i=1}^m\sum_{j=1}' n|a_{ij}|^2}$.}
\end{defn}

The practical turnover is thus determined in terms of the distance of the corresponding rebalancing transform from the identity transform (up-to-permutation). The idea is that the closer one is to the identity transform, the smaller the changes that are applied to the entries of the choice vector.

\begin{prop}
\label{prop:turnover}
Let $\lbold\neq\ubold_n\in\s^n$. For $P\in\Omega_{n,\lbold}= \{ P\in\D_n \mid \ubold_n=\lbold P \}$, denote by $\tilde{\tau}(\lbold) = \{ \tilde{\tau}_P(\lbold) \mid P\in\Omega_{n,\lbold} \}$ the set of all possible practical turnovers. Then
$$ \inf \left( \tilde{\tau}(\lbold) \right)  = \tau(\lbold) \ . $$
In other words, the smallest possible practical turnover is the theoretical turnover.
\end{prop}

\begin{proof}
We will show that $\norm{P-I_n}\geq1$ for all $P\in\Omega_{n,\lbold}$. Note that we obtain the smallest possible norm if all rows of $P$ and $I_n$ coincide up to permutation, except for two rows, say $i$ and $j$. In other words, all entries of $\lbold$ and $\ubold_n$ coincide (up to permutation) apart from the $i$-th and $j$-th entries that need to be averaged out to give $1/n$ each. Because $P$ is a doubly stochastic matrix, the entries of both rows $i$ and $j$ must be some $a\in(0,1)$ and $1-a$. Consequently, $\norm{P-I_n} = \sqrt{2a^2+2(1-a)^2}$ and its minimum is reached at $a=1/2$, implying that the smallest possible norm is equal to $\norm{P-I_n} = \sqrt{4(1/2)^2} = 1$.
\end{proof}

\subsection{Minimal turnover via T-transforms} To characterize the rebalancing transform that would yield the theoretical turnover, and thus by Proposition \ref{prop:turnover} the smallest possible practical turnover, we use the notion of \emph{$T$-transform} (Definition \ref{defn:t_transform}). Recall that in the economic context of equalizing wealth or income, $T$-transforms are also known as \emph{Dalton} or \emph{Robin Hood transfers} and are interpreted as the operation of shifting income or wealth from one individual to a relatively poorer individual.

The following observation follows directly from the proof of Proposition \ref{prop:turnover}.

\begin{cor}
Suppose one can transform $\lbold\in\s^n$ to equality $\ubold_n$ directly through a single $T$-transform, i.e. $T\in\Omega_{n,\lbold}$. Then $\norm{T-I_n} = 1$.
\end{cor}

Also recall that according to \citeasnoun{hardy1934} (Proposition \ref{prop:muirhead}), if a vector $\abold\in\s^n$ is majorized by another vector $\bbold\in\s^n$, then $\abold$ can be derived from $\bbold$ by successive applications of at most $n-1$ such $T$-transforms. Therefore, every rebalancing polytope $\Omega_{n,\lbold}$ contains (not necessarily unique) products of $T$-transforms. In Example \ref{example:rebalance}, $P(0,0)$ is itself a $T$-transform. Such successive applications of $T$-transforms do indeed produce the least possible turnover, that is the theoretical turnover. The following is an immediate consequence of the proof of Proposition \ref{prop:turnover} and the proof of Lemma 2, p.47 of \citeasnoun{hardy1934}.

\begin{prop}
Let $\lbold\neq\ubold_n\in\s^n$. Then
$$ \inf \left( \tilde{\tau}(\lbold) \right)  = \tilde{\tau}_Q(\lbold) \ , $$
where $Q\in\Omega_{n,\lbold}$ is a product of at most $n-1$ $T$-transforms.
\end{prop}

\begin{cor}
For $\lbold\neq\ubold_n\in\s^n$ and the rebalancing polytope $\Omega_{n,\lbold}$, the minimum distance from identity $I_n$ of any rebalancing transform $P\in\Omega_{n,\lbold}$ is a product of $T$-transforms.
\end{cor}

Based on a private correspondence with the authors of \citeasnoun{MarshallOlkin2011}, the problem of characterizing the closest element to an identity matrix within a given polytope has not been tackled in linear algebra. Our characterization through $T$-transforms can hence be of interest to mathematicians and economists working with inequalities and the theory of majorization in general.


\section{Concluding Remarks}
\label{section:conclusion}

We have developed mathematically and economically sound choice theoretic foundations for the naive approach to diversification. In particular, we axiomatized naive diversification by defining it as a preference for equality over inequality and showed that it has a utility representation in terms of Schur-concave functions, which capture the idea of being inequality averse on top of being risk averse. The notion of permutation invariance lies at the core of our naive diversification axiom, since an economic agent with monotonic and continuous preferences is a naive diversifier if and only if his preferences are convex and permutation indifferent. Finally, we showed that the transformations, which rebalance unequal decision weights to equality, are characterized in terms of their implied turnover, and that the least possible turnover is obtained by applying a finite number of $T$-transforms.

We conclude by briefly discussing the relationship between our axiomatic system and observed behavior in reality, followed by sketching some potential choice theoretic extensions of our work.

\subsection{Testing the reality of naive diversification}

Even though desirability for diversification is a cornerstone of a broad range of portfolio choice models, the precise formal definition differs from model to model. Analogously, the way in which the notion of diversification is interpreted and implemented in the real world varies greatly. Traditional diversification paradigms are consistently violated in practice. Indeed, empirical evidence suggests that economic agents often choose diversification schemes other than those implied by Markowitz's portfolio theory or expected utility theory. Diversification heuristics thus span a vast range, and naive diversification, in particular, has been widely documented both empirically and experimentally.

However, despite the growing literature pointing to the common existence of naive diversification in practice, experimental research investigating the \emph{behavioral drivers} of diversifiers remains rather limited. Our axiomatization can help empirical and experimental economists test diversification preferences, and their underlying drivers, of economic agents in the real world. In particular, we can now look for the main parameters driving the decision process of naive diversifiers. One such parameter or heuristic implied by our axiomatization is that of permutation invariance. In practice, it is arguably rather rare that a diversifier would know so little about the given assets to be essentially indifferent among them. Despite this, naive diversification continues to be applied by both experienced professionals and regular people. By varying the amount of information available to subjects in an experimental setting, one may be able to deduce whether the indifference axiom applies in general or whether it is information dependent, as implied by Laplace's principle of indifference. Another insight gained through our axiomatization was that of consistency with traditional convex diversification and concave expected utility maximization.  In particular, consider that a risk averse investor would in theory be expected to diversify in the traditional convex sense. Hence, the level of risk aversion may be yet another parameter driving naive diversification, and this again can be directly tested.

\subsection{Choice-theoretic generalizations}

\paragraph{Comparing allocations among different numbers of choices.}
Our discussion of naive diversification throughout has focused on a fixed number of choice alternatives $n$. Suppose that an economic agent is faced with an allocation among either $\xbold = (x_1,\dots,x_n)$ or $\ybold = (y_1,\dots,y_m)$, where $n\neq m$. In Section 3, we showed that an equal allocation among a larger number of alternatives is always more preferred under naive diversification. More generally, however, given unequal choice weights $\abold = (\alpha_1,\dots,\alpha_n)\in\s^n$ and $\bbold = (\beta_1,\dots,\beta_m)\in\s^m$ and allocations $\abold\cdot\xbold$ and $\bbold\cdot\ybold$, one can cannot infer a preference of one over the other without generalizing the naive diversification axiom. Such an extension has been developed by \citeasnoun{MarshallOlkin2011} in the context of the majorization order on vectors of unequal lengths. In fact, they showed that the components of $\abold$ are less spread out than the components of $\bbold$ if and only if the Lorenz curve $L_{\abold}$ associated with the vector $\abold$ is greater or equal than the Lorenz curve $L_{\bbold}$ associated with $\bbold$ for all values in its domain $[0,1]$, and that this is equivalent to requiring that $1/n\sum_{i=1}^n \phi(\alpha_i) \leq 1/m \sum_{i=1}^m \phi(\beta_i)$ for all convex functions $\phi:\R\to\R$.

\paragraph{Multidimensional diversification.}
One may think of naive diversification as being \emph{univariate}, in the sense that a naive diversifier is concerned with only one dimension, namely that of equality of choice weights. Suppose that an economic agent would like to diversify naively, but would also like to reduce variability along a second dimension. Consider for example the dimension of ``risk weights" as opposed to ``capital weights". This is a commonly applied risk diversification strategy in practice, known under \emph{risk parity}. Parity diversification focuses on allocation of risk, usually defined as volatility, rather than allocation of capital. Here, risk contributions across choice alternatives are equalized (and are in practice typically levered to match market levels of risk). It can be viewed as a middle ground between the naive approach and the minimum risk approach (see for example \citeasnoun{Roncalli2010}).

When allocations along more than one dimension are to be compared simultaneously, we move from the linear space of choice vectors to the space of choice \emph{matrices}. Each row of a choice matrix represents a particular attribute or dimension, whereas each column represents the choice weights along that dimension. The generalization of the mathematical formalism of naive diversification is then straightforward. For example, a choice matrix $X$ is more diversified (along some given dimensions) than a choice matrix $Y$ if $X=PY$ for some doubly stochastic matrix $P$. This definition is part of an established field within linear algebra known as multivariate majorization.

\paragraph{Towards an inequality aversion coefficient.}
The naive diversification axiom implies that a weight allocation that is closest to the equal weighted vector $\ubold_n$ is always more preferred. This in turn induces the idea of being \emph{averse to inequality}, which we discussed in Section 4. One may formalize this notion, together with a characterization of different levels of inequality aversion as follows.

First, yet another generalization of naive diversification can be obtained by substituting a more general vector $\dbold\in\s^n$ for the equality vector $\ubold_n$. In that case, weight allocations closest to $\dbold$ are preferred. To do this, we need to define the concept of $d$-stochastic matrix. For $\dbold\in\s^n$, an $n\times n$ matrix $A=(a_{ij})$ is said to be \emph{$\dbold$-stochastic} if (i) $a_{ij}\geq0$ for all $i,j\leq n$; (ii) $\dbold A = \dbold$; and (iii) $A\ubold'_n = \ubold'_n$. To get an intuition for $\dbold$-stochastic matrices, note that since $\sum_{i=1}^n d_i = 1$ by construction, a $\dbold$-stochastic matrix in our setting can be viewed as the transition matrix of a Markov chain.
Clearly, when $\dbold=\ubold_n$, a $\dbold$-stochastic matrix is doubly stochastic.
One can then say that a preference relation $\succsim$ exhibits preference for \emph{relative naive diversification} if there is a weight allocation $\dbold = (d_1,\dots,d_n) \in\s^n$ such that for any $\abold=(\alpha_1,\dots,\alpha_n)\in\s^n$ and $\bbold = (\beta_1, \dots, \beta_n) \in \s^n$,
$$ \abold \le_m\bbold \iff \abold = \bbold A  $$
for some $\dbold$-stochastic matrix $A$.
The interpretation here is that an individual with naive diversification preferences relative to some $\dbold\neq\ubold_n$ is \emph{less averse} to inequality than one with naive diversification preferences.

To be then able to compare levels of aversion to inequality within relative naive diversification preferences, we can introduce the \emph{coefficient of inequality aversion}.
For naive diversification preferences relative to $\dbold\in\s^n$, the corresponding \emph{inequality aversion coefficient} $\ineq$ is defined as $\ineq = \norm{\dbold-\ubold_n}$, where $\norm{\cdot}$ is the Euclidean norm taken up-to-permutation.
Clearly, this inequality aversion coefficient $\ineq$ lies within $[0,\infty)$, with $\ineq=0$ for naive diversification preferences, in which case we can say that the decision maker possesses \emph{absolute aversion to inequality}.

\bibliographystyle{econometrica}
\bibliography{refDB,divreview_references}

\end{document}